\documentclass[11pt]{article}

\usepackage{array}
\usepackage{amsfonts,amsmath,amssymb,amsthm,mathrsfs}
\usepackage{adjustbox}
\usepackage{tikz}
\usepackage{enumitem}
\usepackage{ifthen}

\usepackage{fullpage}

\usepackage[ruled]{algorithm2e} 

\SetAlFnt{\small}
\SetAlCapFnt{\small}
\SetAlCapNameFnt{\small}
\SetAlCapHSkip{0pt}
\IncMargin{-\parindent}

\usepackage{booktabs} 
\usepackage{xspace}
\usepackage[pagebackref=true,hidelinks]{hyperref}
\usepackage{bbm}
\usepackage{enumitem}
\usepackage[numbers]{natbib}
\usepackage{multirow}
\usepackage{graphicx}
\usepackage{subfigure}
\usepackage{wrapfig}
\usepackage{cleveref}

\newtheorem{theorem}{Theorem}[section]

\newtheorem{lemma}[theorem]{Lemma}
\newtheorem{corollary}[theorem]{Corollary}

\newenvironment{numberedtheorem}[1]{%
\renewcommand{\thetheorem}{#1}%
\begin{theorem}}{\end{theorem}\addtocounter{theorem}{-1}}

\newcommand{\rev}{\textsc{Rev}}
\newcommand{\srev}{\textsc{SRev}}

\newcommand{\E}{\mathbb{E}}

\newcommand{\adaptrev}{\textsc{BuyManyRev}}
\newcommand{\buymanyrev}{\textsc{BuyManyRev}}

\newcommand{\cs}{c}

\newcommand{\zeros}{\mathbf{0}}

\newcommand{\udelta}{\Delta_m}
\newcommand{\exantesrev}{\textsc{EA-SRev}}
\newcommand{\exantebuymanyrev}{\textsc{EA-BuyManyRev}}
\newcommand{\ps}{\mathbf{p}}

\newcommand{\dists}{\mathfrak{D}}
\newcommand{\profit}{\textsc{Profit}}
\newcommand{\sprofit}{\textsc{SProfit}}
\newcommand{\buymanyprofit}{\textsc{BuyManyProfit}}
\newcommand{\alloc}{\mathbf{x}}

\newcommand{\R}{\mathbb{R}}

\newcommand{\dist}{\mathcal{D}}

\newcommand{\argmax}{\operatorname{arg\,max}}
\newcommand{\argmin}{\operatorname{arg\,min}}

\newcommand{\remove}[1]{}

\newcommand{\M}{\textsc{M}}

\begin{document}

\title{Buy-Many Mechanisms for Many Unit-Demand Buyers}
\author{
Shuchi Chawla \\ UT-Austin \\ {\tt shuchi@cs.utexas.edu} \and 
Rojin Rezvan \\ UT-Austin\\ {\tt rojinrezvan@utexas.edu} \and
Yifeng Teng \\ Google Research \\ {\tt yifengt@google.com} \and 
Christos Tzamos \\ UW-Madison \\ {\tt tzamos@wisc.edu}
}
\date{}

\maketitle
\thispagestyle{empty}

\begin{abstract}
A recent line of research has established a novel desideratum for designing approximately-revenue-optimal multi-item mechanisms, namely the buy-many constraint. Under this constraint, prices for different allocations made by the mechanism must be subadditive, implying that the price of a bundle cannot exceed the sum of prices of individual items it contains. This natural constraint has enabled several positive results in multi-item mechanism design bypassing well-established impossibility results. Our work addresses the main open question from this literature of extending the buy-many constraint to multiple buyer settings and developing an approximation.

We propose a new revenue benchmark for multi-buyer mechanisms via an ex-ante relaxation that captures several different ways of extending the buy-many constraint to the multi-buyer setting. Our main result is that a simple sequential item pricing mechanism with buyer-specific prices can achieve an $O(\log m)$ approximation to this revenue benchmark when all buyers have unit-demand or additive preferences over m items. This is the best possible as it directly matches the previous results for the single-buyer setting where no simple mechanism can obtain a better approximation. 

From a technical viewpoint we make two novel contributions. First, we develop a supply-constrained version of buy-many approximation for a single buyer. Second, we develop a multi-dimensional online contention resolution scheme for unit-demand buyers that may be of independent interest in mechanism design. 



\end{abstract}

\begin{titlepage}

\maketitle

\end{titlepage}

\section{Introduction}

Revenue maximization in multi-parameter settings is notoriously challenging. It is known, for example, that in the absence of strong assumptions on the buyer's value distribution, the optimal revenue cannot be approximated within any finite factor by any mechanism with finite description complexity even for the simplest possible setting of two items and a single unit-demand buyer \cite{hart2013menu, briest2015pricing}. These impossibility results motivate the search for a different benchmark that captures salient features of the problem space while also permitting non-trivial approximation. 

For single-agent settings one such benchmark was proposed by \citet*{briest2015pricing} and \citet*{chawla2019buy}. These works showed that the infinite gap between the revenue of  the optimal mechanism and any simple mechanism arises precisely when the mechanism offers ``super-additive'' options that charge for some bundles of items more than the sum of prices of their individual components. 
The so-called ``buy-many'' constraint disallows such exploitative behavior: when interpreted as a pricing over all possible allocations, the mechanism should be subadditive (defined appropriately over randomized allocations, as described in Section~\ref{sec:prelim}). An alternate motivation for this constraint arises as a consequence of buyer behavior in scenarios where the buyer can interact with the mechanism multiple times, purchasing an option from the menu each time. The buyer can then construct an allocation by purchasing a multiset of options from the menu in the cheapest possible manner. The buy-many constraint restricts the kinds of mechanisms the seller can use to extract revenue from the buyer. The corresponding benchmark is the optimal revenue that can be obtained from any mechanism satisfying the constraint. \cite{briest2015pricing} and \cite{chawla2019buy} showed that imposing such a constraint enables positive results even without requiring any assumptions on the buyer's valuation function or the value distribution: for settings with a single buyer and $m$ items, the optimal revenue obtained from any mechanism satisfying the buy-many constraint is no more than $O(\log m)$ times the revenue obtained from an item pricing.

Our work addresses the primary direction left open by these works and their followups, namely extending the buy-many constraint and revenue benchmark to settings with multiple buyers. Obtaining such an extension is challenging, however, as it depends on the particular implementation of the mechanism, differences in these details can lead to very different benchmarks. Indeed while the two approaches described above for formalizing the buy-many constraint -- as a restriction on the pricing function and as a consequence of buyer behavior -- lead to equivalent definitions in the single-buyer setting, they turn out to be very different in the multi-buyer setting. In fact, the latter approach of allowing buyers to interact with the mechanism multiple times provides no meaningful restriction on mechanisms at all and once again allows for unbounded revenue gaps between simple and optimal mechanisms.\footnote{Consider, for example, multiple interactions of a buyer with the mechanism interleaved by purchases made by other buyers. As the item supply changes, the mechanism can update the prices on its menu, and no longer necessarily needs to satisfy a subadditivity constraint on the final pricing observed by the buyer. In fact, by exploiting this supply-based pricing approach, a multi-buyer buy-many mechanism can simulate any single-agent non-buy-many mechanism, inheriting the unbounded simple-versus-optimal revenue gaps of the latter setting. Sybil-proofness or false-name-proofness is even easier to achieve in principle, unless some symmetry-type restrictions are placed on the mechanism (as in \cite{Makoto04, IwasakiCOSTGY10}, for example), as the mechanism can simply refuse to make any allocations unless the number of agents is exactly $n$.} We instead define the buy-many constraint as a restriction on the effective pricing faced by any individual buyer participating in the mechanism. We view the buy-many constraint as a standalone desideratum for the design of mechanisms that disallows exploitative behavior on part of the seller. This approach leads to a non-trivial restriction, while continuing to permit a range of different designs. We show that it becomes possible to recover upper bounds on the revenue gap between simple and optimal mechanisms without assumptions about buyers' value distributions, just as in the single-buyer setting.

Even as a restriction on the mechanism's pricing, the buy-many constraint can take many different forms depending on the information available to buyers at different stages in the mechanism. To obtain a comprehensive theory of multi-buyer buy-many mechanisms without going into the intricacies of specific applications, we avoid choosing any particular such extension, and instead define an ex-ante relaxation of the optimal buy-many revenue that simultaneously captures a broad range of settings. In Section~\ref{sec:prelim}, we describe the different forms of the buy-many constraint encompassed by this relaxation.

\subsection*{An ex-ante relaxation of buy-many revenue}

We view multi-buyer settings as a collection of single-buyer instances via an ex-ante relaxation that allows the mechanism to allocate each item multiple times as long as the {\em expected} number of buyers each item is allocated to is at most 1. To be specific, let $x_i$ be an $m$-dimensional allocation vector, with $x_{ij}\in [0,1]$ denoting the probability of allocation of item $j$ to a particular buyer $i$. We consider single-buyer mechanisms satisfying two restrictions. First, over the randomness in the buyer's valuation function, we require that the probability of allocation of each item $j$ to the buyer is at most $x_{ij}$. We say that the mechanism satisfies the ex-ante constraint $x_i$. Second, we require that the mechanism satisfies the single-buyer buy-many constraint. We then consider the maximum revenue that can be obtained from any mechanism for buyer $i$ that satisfies both of the aforementioned constraints. Note that the buy-many constraint is not closed over convex combinations, so distributions over buy-many menus can potentially obtain higher revenue than buy-many menus themselves. We accordingly consider the maximum revenue obtainable from individual buy-many mechanisms or distributions over buy-many mechanisms that satisfy the ex-ante constraint $x_i$. Let $\textsc{BuyManyRev}_i(x_i)$ denote this upper bound. 
The following program then gives an upper bound on the revenue of multi-buyer buy-many mechanisms:
$$ \textsc{ExAnte-BuyManyRev} \, :=\, \max_{x_1,...,x_n \ge \mathbf{0}} \sum_i \textsc{BuyManyRev}_i(x_i) \text{ such that } \sum_i x_i \preceq \mathbf{1}. \footnote{For any two vectors $x,y\in \R^m$, $y\preceq x$ means $y_j\leq x_j$ for every $j\in[m]$.}$$ 

Our goal in this work is to design simple multi-buyer mechanisms that are clearly buy-many but at the same time are competitive with this new benchmark.



\subsection*{Approximating the buy-many revenue by item pricings}

Although buy-many mechanisms can, in general, be quite complicated, we show that when all buyers are unit-demand or additive, the optimal buy-many revenue (as well as the ex-ante upper-bound) can be approximated via a simple class of mechanisms, namely sequential item pricings. We obtain an $O(\log m)$ approximation where $m$ is the number of items, matching within constant factors the approximation achieved by \citet{chawla2019buy} and \citet{briest2015pricing} for single buyer settings. 

\begin{theorem}\label{thm:main}
For $n$ independent buyers that are unit-demand or additive over $m$ items with value distribution $\dists$,
\begin{equation*}
    \textsc{ExAnte-BuyManyRev}(\dists)\leq O(\log m)\srev(\dists).
\end{equation*}
\end{theorem}

Here $\srev$ denotes the optimal revenue achievable through sequential item pricings.
A sequential item pricing mechanism interacts with buyers sequentially in a particular order. It offers to each buyer the set of remaining items at predetermined item prices and allows the buyer to purchase any item of her choice. Notably, our approximation result holds for a worst-case order of arrival of the buyers. Furthermore, the item prices offered to each buyer are non-adaptive in that prices are computed once at the beginning of the mechanism before buyer values are instantiated, and do not change based on instantiated values and purchasing decisions of buyers that arrive earlier in the ordering. Interestingly, sequential item pricing was previously shown in \cite{chawla2010multi} to obtain a constant-factor approximation to the overall (non-buy-many) optimal revenue for unit-demand buyers when buyers' values are independent across items.


\subsection*{Our techniques}

Our main approximation result consists of two parts, each of which is of independent interest. First, we define an ex-ante supply-constrained relaxation of (distributions over) item pricings, similar to the ex-ante relaxation of general buy-many mechanisms discussed above. We show that the ex-ante item pricing revenue provides a logarithmic approximation to the ex-ante buy-many revenue for buyers with any combinatorial valuation function.

Focusing on ex-ante relaxations allows us to revert back to the single-agent setting. We extend the logarithmic upper bound on the gap between buy-many and item pricing revenues proved by \cite{chawla2019buy} to the single buyer setting with a supply constraint. However, extending \cite{chawla2019buy}'s argument is not straightforward because it does not provide any control on how the expected allocation of the item pricing relates to that of the optimal buy-many mechanism. We instead consider a Lagrangian version of the limited supply setting: in this setting, we are given a production cost for each item and our goal is to maximize the mechanism's profit, namely its revenue minus the expected cost it has to pay to produce the items it sells. Unfortunately, it turns out that for some cost vectors, the expected profit of a buy-many mechanism can be an $\Omega(m)$ factor larger than the expected profit of any item pricing. We show instead that the optimal item pricing profit approximates the profit of an optimal buy-many mechanism that faces slightly higher (in multiplicative terms) production costs. This allows us to obtain a logarithmic bound on the gap between the two ex-ante relaxations.

\begin{theorem}\label{cor:eabuymany-vs-srev}
For any joint value distribution $\dists$ over $n$ buyers and $m$ items, 
\begin{equation*}
    \textsc{ExAnte-BuyManyRev}(\dists)\leq O(\log m) \textsc{ExAnte-SRev}(\dists).
\end{equation*}
\end{theorem}

{\bf A multi-dimensional contention resolution scheme.} The second part of our argument relates the ex-ante item pricing revenue to the revenue of a non-adaptive sequential item pricing for unit-demand or additive buyers. This part employs an argument reminiscent of prophet inequalities and contention resolutions schemes (see, e.g., \citet*{feldman2016online}). Specifically, let $x_i$ denote the allocation vector for agent $i$ in the ex-ante-optimal item pricing revenue. For any given fixed ordering over the buyers, we construct pricings $\{q_i\}$ such that: (1) every buyer $i$ obtains an expected allocation of at least $x_i/2$, and (2) the revenue obtained by the item pricing $q_i$ from agent $i$ is at least half the agent's contribution to the ex-ante item pricing revenue. Both of these properties are easy to observe for additive buyers, but challenging for unit-demand buyers. The challenge in showing (1) is that the set of items available when agent $i$ arrives depends upon the instantiations of previous agents' values: as this set changes, the buyer's choice of what to buy also changes. (2) is challenging because revenue is not linear in allocation probability. Our key observation is that for unit-demand buyers, item pricing revenue exhibits concavity as a function of allocation probabilities: given any item pricing $p$ with allocation probabilities $x$ and any $\alpha<1$, we can find another item pricing that allocates items with probabilities at most $\alpha x$ and obtains at least an $\alpha$ fraction of $\ps$'s revenue. We leave open the question of extending this type of multi-dimensional contention resolution scheme (and in particular, achieving the following theorem) to other valuation functions.

\begin{theorem}\label{thm:seq-pricing-approx-exante}
For any joint value distribution $\dists$ over $n$ unit-demand or additive buyers and $m$ items, 
\begin{equation*}
    \textsc{ExAnte-SRev}(\dists)\leq 2 \srev(\dists).
\end{equation*}
\end{theorem}

Theorem~\ref{thm:main} follows immediately from Theorems~\ref{cor:eabuymany-vs-srev} and \ref{thm:seq-pricing-approx-exante}. We prove Theorem~\ref{cor:eabuymany-vs-srev} in Section~\ref{sec:exante} and Theorem~\ref{thm:seq-pricing-approx-exante} in Section~\ref{sec:approx}.

In Section~\ref{discussion}, we mention some \textbf{motivating examples} regarding multi-buyer buy-many mechanisms, and address what are the main difficulties of extending our results to all valuation functions.



\subsection*{Further related work}

The buy-many constraint was first proposed as an alternative to unconstrained revenue maximization by \citet*{briest2015pricing} in the context of a single unit-demand buyer. \citet*{chawla2019buy} extended the notion to arbitrary single buyer settings. \citet*{chawla2021pricing} present improved approximations for buy-many mechanisms in single buyer settings where buyers' values satisfy an additional ordering property that includes, for example, the so-called FedEx setting. In a different direction, \citet*{chawla2020menu} study the menu size complexity and revenue continuity of single buyer buy-many mechanisms. All of these works focus on single buyer problems.

There is extensive literature on approximating the optimal non-buy-many revenue under various constraints on the buyers' value distributions in both single-buyer and multi-buyer settings. Interestingly, many of these approximation results are achieved through sequential item pricing (or, in some cases, grand bundle pricing) mechanisms. For example, \citet*{chawla2010multi} show that for multiple settings of matroid feasibility constraints, sequential item pricings give a constant approximation to the optimal non-buy-many revenue. In particular, for multiple unit-demand buyers with independent values over all items, sequential item pricing gives a 6.75-approximation to the optimal revenue, and this competitive ratio was improved to 4 by \citet*{alaei2014bayesian} and generalized to a setting where each item has multiple units. For buyers with more general distribution, either a sequential pricing mechanism with personalized item prices or a sequential item pricing with anonymous prices and entry fee gives a constant approximation to the optimal revenue when there are multiple fractional subadditive buyers \cite{cai2017simple} and an $O(\log \log m)$ approximation when buyers are subadditive \cite{duettinglog,cai2017simple}. \citet*{MaS21} consider additive-valued buyers and a seller facing production costs, a setting that we revisit in our proof of Theorem~\ref{thm:constrained-alloc-approx}, and achieve approximations using two-part tariffs. Also see \cite{chawla2007algorithmic,chawla2015power,kleinberg2012matroid,hart2013menu,li2013revenue,babaioff2014simple,yao2014n,rubinstein2018simple} for some other previous work on using simple mechanisms to approximate the optimal non-buy-many revenue. For sequential item pricings to be able to approximate the revenue of the optimal non-buy-many mechanisms, an important assumption is that each buyer should have independent item values, which we do not assume in this work.

Finally, ex-ante relaxations, first introduced to multi-buyer mechanism design by \citet{alaei2014bayesian} and \citet{yan2011mechanism}, have emerged as a powerful tool for simplifying multi-buyer mechanism design problems by breaking them up into their single-agent counterparts. For example, \citet{chawla2016mechanism} and \citet{cai2017simple} employ this approach for designing sequential mechanisms that approximate the optimal revenue for multiple agents with subadditive values. Such a technique is also used for analyzing the setting where buyers have non-linear utilities \cite{alaei2013simple,feng2019optimal,feng2020simple}, and determining the revelation gap of the optimal mechanisms \cite{feng2018end,feng2021revelation}. 


\section{Definitions}
\label{sec:prelim}

We study the multidimensional mechanism design problem where the seller has $m$ heterogeneous items to sell to $n$ buyers with independent value distributions, and aims to maximize the revenue. The buyer $i$'s value vector over items is specified by a distribution $\dist_i$ over all value functions $v_i:2^{[m]}\rightarrow \mathbb{R}^{\geq 0}$. Each of these such functions $v_i$ assigns a non-negative value to any subset of items $S\subseteq [m]$. 
Let $\dists=\dist_1\times\cdots\times\dist_n$ denote the joint distribution over all values. Let $\Delta_m = [0,1]^m$ be the set of all possible randomized allocations.

\paragraph{Unit-demand and additive Buyers} In this paper, we focus on settings where every buyer is unit-demand or where every buyer is additive valued. We say that a buyer is unit-demand over all items, if the buyer is only interested in purchasing one item, and her value for any set of items is solely determined by the item that is most valuable to her. In other words, for any set of items $S \subseteq [m]$, $v_i(S) = \max\limits_{j\in S} v_i(\{j\})$. We say that a buyer is additive if for all $S \subseteq [m]$, $v_i(S) = \sum\limits_{j\in S} v_i(\{j\})$. 
For ease of notation, let $v_i(\{j\})=v_{ij}$. 
We make no assumptions on each buyer's value distribution $\dist_i$ -- the buyer's values for different items can be arbitrarily correlated.

\paragraph{Single-buyer Mechanisms} By the taxation principle, any single-buyer mechanism is equivalent to a menu of options, and the buyer can select an outcome that maximizes her utility. Without loss of generality, we can assume that the menu assigns a price to {\em every} randomized allocation, a.k.a. lottery, $\lambda\in\udelta$, and can therefore simply represent the mechanism by a pricing function $p$. $p(\lambda)$ denotes the price of the lottery $\lambda$.
For a buyer of type $v_i$, her (expected) value for a lottery $\lambda$ is defined by $v_i(\lambda):=\sum_{j}v_{ij}\lambda_j$, and her utility for the lottery is defined to be $u_{v_i,p}(\lambda):=v_i(\lambda)-p(\lambda)$. When $p$ is clear from the context, we drop the subscript and write the buyer's utility as $u_{v_i}(\lambda)$. Henceforth, we will refer interchangeably to mechanisms as pricing functions.

Given a pricing function $p$, a buyer of type $v$ chooses the utility-maximizing lottery, denoted $\lambda_{v,p} := \argmax_\lambda u_{v,p}(\lambda)$. 
The buyer's utility in the mechanism $p$ is given by $u_p(v):=v(\lambda_{v,p})-p(\lambda_{v,p})$; the buyer's payment is $\rev_p(v):=p(\lambda_{v,p})$. We write the revenue of the mechanism under buyer distribution $\dist$ as $\rev_{p}(\dist)=\E_{v\sim\dist}\rev_p(v)$. The mechanism is called a buy-one mechanism since the buyer only purchases one option from the menu. 

\paragraph{Single-buyer Buy-Many Mechanisms} In single buyer settings, if the buyer is allowed to purchase multiple options from a mechanism's menu, we call the mechanism buy-many. In particular, in a buy-many mechanism the buyer can purchase a (random) sequence of lotteries, where each lottery in the sequence can depend adaptively on the instantiations of previous lotteries. At the end of the process, the buyer gets the union of all allocated items in each step and pays the sum of the prices of all purchased lotteries. 


Any buy-many mechanism can be described by a buy-one pricing function that satisfies a buy-many constraint. Intuitively, a buy-one pricing function satisfies the buy-many constraint if the buyer always prefers to purchase a single option from the menu, even if she has the option to adaptively interact with the mechanism multiple times.\footnote{For formal definitions, see \cite{chawla2019buy}.} For example, for a pricing function $p$, let $\hat{p}_i$ denote the minimum cost of acquiring item $i$ by repeatedly purchasing some lottery until the item is instantiated. Then, the buy-many constraint implies that for any $\lambda$, $p(\lambda)\le \hat{p}\cdot \lambda$. Pricing functions that satisfy the buy-many constraint are called buy-many pricings, and we use $\textsc{BuyMany}$ to denote the set of such functions. 


Let $\buymanyrev(\dist):=\max_{p\in\textsc{BuyMany}} \rev_p(\dist)$ denote the revenue of the optimal buy-many mechanism for a buyer with value distribution $\dist$.


\paragraph{Item Pricings and Sequential Item Pricings} In a single-buyer setting, a (deterministic) item pricing mechanism sets an item price $p_j\in\R_+$ for every item $j$; a buyer with value function $v_i$ purchases the item $j$ that maximizes $v_{ij}-p_j$. We use the price vector $p=(p_1,\cdots,p_m)$ to denote the item pricing. Denote by $\srev(\dist)$ the optimal revenue obtained by any item pricing for a buyer with distribution $\dist$. We also define random item pricing mechanisms as distributions over deterministic item pricings; we continue to use $p$ to denote such pricings, although $p$ is now a random variable. 

A sequential item pricing for a multi-buyer setting specifies a serving order $\sigma$ being a permutation over the $n$ buyers, and $n$ item pricings $p_1,\cdots,p_n$, such that at step $i$ the seller posts the pricing $p_{\sigma(i)}$ to buyer $\sigma(i)$, and the buyer purchases her favorite among the items that are still available. We use $\srev(\dists)$ to denote the optimal revenue obtained by a sequential item pricing mechanism for buyers with values drawn from the joint distribution $\dists$.

\paragraph{Profit Maximization under Production Costs} For a single-buyer setting, we study the profit-maximizing problem of the seller where each item has a production cost. Let $\cs=(c_1,c_2,\cdots,c_m)\in\R_+^m$ denote the vector of production costs. 
Then the profit of the single-buyer mechanism given by pricing $p$ is defined to be the revenue minus the production costs of the items:
\begin{equation*}
\profit_{p,\cs}(\dist) = \E_{v\sim\dist}[p(\lambda_{v,p})-\lambda_{v,p}\cdot\cs].
\end{equation*}
We denote by $\sprofit_{\cs}(\dist)=\max_{\textrm{item pricing }p}\profit_{p,\cs}(\dist)$ the optimal profit achievable by item pricings for the buyer with distribution $\dist$ when there are production costs $\cs$ for all items. Similarly define $\buymanyprofit_{\cs}(\dist)=\max_{\textrm{buy-many }p}\profit_{p,\cs}(\dist)$ to be the optimal profit achievable by buy-many mechanisms for the buyer with distribution $\dist$ for the setting with production costs $\cs$.


\subsection*{Ex-ante constrained revenue}

Ex-ante relaxations are a powerful technique for reducing multi-buyer mechanism design problems to their single-buyer counterparts. The key idea is to relax the ex-post supply constraint on items to an ex-ante feasibility constraint that requires each item to be sold at most once in expectation. 

Recall that $x_{ij}$ denotes the probability of allocating item $j$ to buyer $i$; $x_i = (x_{i1}, \cdots, x_{im})$ denotes the vector of allocations of all items to buyer $i$; and $x = (x_1, \cdots, x_n)$ denote the vector of all allocations.


We say that a single-buyer pricing function $p_i$ satisfies the ex-ante constraint $x_i$ with respect to the value distribution $\dist_i$ if for all $j\in [m]$, the expected allocation of item $j$ to the buyer is no more than $x_{ij}$: $\E_{v\sim\dist_i}[\lambda_{v,p_i}]\le x_i$. Note that this definition extends to random pricing functions $p$ in the straightforward manner: we want the expected allocation to be bounded by $x_i$ where the expectation is taken over both the buyer's value and the randomness in the mechanism. That is, $\E_{v\sim\dist_i, p}[\lambda_{v,p_i}]\le x_i$. 

For a single buyer $i$ and an ex-ante constraint $x_i$, we can now define the optimal revenue that can be obtained from the buyer subject to an ex-ante constraint for various classes of mechanisms. In particular, let $\Delta_{IP}$ denote the space of all distributions over item pricings, and $\Delta_{BM}$ denote the space of all distributions over buy-many pricings. Then we define:
\begin{align*}
    \srev(\dist_i, x_i) & = \max_{p\in \Delta_{IP}:\, p \text{ satisfies $x_i$ w.r.t. $\dist_i$}} \rev_p(\dist_i), & \text{and, }\\
    \adaptrev(\dist_i, x_i) & = \max_{p\in \Delta_{BM}:\, p \text{ satisfies $x_i$ w.r.t. $\dist_i$}} \rev_p(\dist_i). &
\end{align*}
\noindent
Given a combined vector $x$ of ex-ante constraints for every buyer $i\in [m]$, we can write
the collective revenue of the optimal single-buyer mechanisms that satisfy these constraints as:
\begin{align*}
    \exantesrev(\dists,x) & = \sum_i \srev(\dist_i, x_i), & \text{and, }\\
    \exantebuymanyrev(\dists,x) & = \sum_i \buymanyrev(\dist_i, x_i),
\end{align*}
\noindent
Finally, we can define the ex-ante relaxation for each class of mechanisms. 
\begin{align*}
    \exantesrev(\dists) & = \max_{x:\, \sum_i x_{ij}\le 1 \forall j\in [m]} \exantesrev(\dists, x), & \text{and, }\\
    \exantebuymanyrev(\dists) & = \max_{x:\, \sum_i x_{ij}\le 1 \forall j\in [m]} \exantebuymanyrev(\dists, x),
\end{align*}

\subsection*{Some settings captured by the ex-ante relaxation}
Our ex-ante relaxation captures approaches that define the buy-many constraint as a restriction on the effective pricing faced by any individual buyer participating in the mechanism. We now describe some specific such settings. We start with one of the simplest settings, in which buyers arrive one after the other and each buyer faces a single-buyer mechanism.

\paragraph{Sequential Buy-Many Mechanisms}  Sequential mechanisms make offers to each buyer sequentially in a pre-specified order. The $i$-th buyer in the order is asked to choose which items to purchase among the subset of items that remain after buyers $1$ through $i-1$ have made their choices. A natural definition for the buy-many constraint for this multi-buyer setting  boils down to offering a buy-many constrained mechanism to each buyer for any subset of items remaining, i.e. the full-mechanism can be defined as a collection of single-buyer mechanisms $\mathcal{M}_{i,S} \in \textsc{BuyMany}$ for any buyer $i$ and any subset $S$ of remaining items.


\paragraph{Ex-Post Buy-Many Mechanisms} A more flexible design space for multi-buyer buy-many mechanisms is to consider direct mechanisms in which buyers truthfully submit their valuations and conditional on the valuations of other buyers each buyer is faced with a single buyer buy-many mechanism. Buy-many mechanisms for this setting are specified as a collection of single buyer buy-many mechanisms $\mathcal{M}_{i,\vec{v}_{-i}} \in \textsc{BuyMany}$ for any buyer $i$ and any combination of valuation functions $v_{-i}$ for the other buyers. 

\paragraph{Bayesian Buy-Many Mechanisms} Another potential definition of multi-buyer buy-many mechanisms is to consider Bayesian settings in which we require the options any single buyer faces to be buy-many in expectation over the valuations of other buyers. If $q_i(v_i, \vec{v}_{-i})$ and $p_i(v_i, \vec{v}_{-i})$ is the  allocation and price offered to buyer $i$ when her value is $v_i$ and the other buyers have valuations $\vec{v}_{-i}$, we would require that the single buyer mechanism with allocation probabilities $q_i(v_i) \triangleq \mathbb{E}_{\vec{v}_{-i}} q_i(v_i, \vec{v}_{-i})$ and price $p_i(v_i) \triangleq \mathbb{E}_{\vec{v}_{-i}} p_i(v_i, \vec{v}_{-i})$ to satisfy the buy-many constraint.

\section{Relating the Ex-Ante Relaxations}
\label{sec:exante}

In this section we bound the gap between the ex-ante optimal buy many revenue and the ex-ante optimal item pricing revenue when the seller faces many buyers. This is a generalization of single-buyer buy-many revenue approximations to supply constrained settings. We emphasize that for the results in this section we {\em do not} require any assumptions on the buyer's valuation function, such as that it is unit-demand or additive. 

A note on notation: since we consider a single-buyer problem in this section, we drop the subscript $i$ from most notation and simply denote buyer $i$'s valuation function as $v$; her allocation vector as $x$; the probability that item $j$ is allocated to the buyer as $x_j$; etc. We will write the ex-ante buy many and item pricing revenues of this buyer simply as $\buymanyrev(\dist, x)$ and $\srev(\dist, x)$ respectively.


\begin{theorem}\label{thm:constrained-alloc-approx}
For any single buyer with value distribution $\dist$ over $m$ items and any ex-ante supply constraint $x\in\udelta$, 
\begin{equation*}
    \exantebuymanyrev(\dist,x)\leq O(\log m) \exantesrev(\dist,x).
\end{equation*}
\end{theorem}

\noindent
Applying this theorem to each of $n$ buyers, we obtain the following corollary.
\begin{numberedtheorem}{\ref{cor:eabuymany-vs-srev}}
For any joint value distribution $\dists$ over $n$ buyers and $m$ items, 
\begin{equation*}
    \exantebuymanyrev(\dists)\leq O(\log m) \exantesrev(\dists).
\end{equation*}
\end{numberedtheorem}

Before we present a formal proof of Theorem~\ref{thm:constrained-alloc-approx}, let us describe the main ideas. Theorem~\ref{thm:constrained-alloc-approx} is a generalization of Theorem 1.1 from \citet{chawla2019buy} \remove{for unit-demand buyers}, which shows that the ratio of $\buymanyrev$ and $\srev$ is bounded by $O(\log m)$ in the absence of an ex-ante supply constraint. The proof technique of \citeauthor{chawla2019buy} does not directly lend itself to the ex-ante setting because it does not provide much control over the allocation probability of the random item pricing it produces. Indeed, the (random) item pricing it returns is independent of the buyer's value distribution, whereas the allocation probabilities (that are expectations over values drawn from the distribution) necessarily depend on the value distribution. Instead of applying the approach of \citeauthor{chawla2019buy} directly, we first consider a Lagrangian version of the supply constrained problem. 


To simplify the following discussion, we will hide the argument $\dist$ from the respective revenue benchmarks. Viewing $\srev(x)$ and $\buymanyrev(x)$ as two multi-variate functions over $x\in\udelta$, we first observe that for any fixed ex-ante constraint $x^o$, there exists a cost vector $\cs^o$ such that $x^o$ is the solution to the optimization problem $\max_x (\srev(x)-\cs^o\cdot x)$. Indeed, because $\srev(x)$ is concave\footnote{In fact, the function $\rev(x)$ defined as maximum revenue from any restricted set of mechanism with allocations at most $x$ is concave. This is because for any two allocations $x$ and $y$ and coefficient $1\geq \alpha \geq 0$, one can consider mechanisms $\M(x)$ and $\M(y)$ that define $\rev(x)$ and $\rev(y)$, and run the former with probability $\alpha$ and the latter with probability $(1-\alpha)$. Then $\rev(\alpha x+(1-\alpha)y)\geq \alpha \rev(x)+(1-\alpha)\rev(y)$.}, $\cs^o=\nabla\srev(x^o)$ is such a function. Furthermore, $\cs^o$ is a non-negative vector, and so the gap between $\srev(x^o)$ and $\buymanyrev(x^o)$ is bounded by the gap between $\srev(x^o)-\cs\cdot x^o$ and $\buymanyrev(x^o)-\cs\cdot x^o$. This motivates studying the Lagrangian problem of maximizing the profit of a mechanism subject to production costs $\cs^o$. In particular, for any value distribution $\dist$, we have:
\[
\max_{x} \frac{\exantebuymanyrev(\dist,x)}{\exantesrev(\dist,x)} \le \max_{\cs} \frac{\buymanyprofit_{\cs}(\dist)}{\sprofit_{\cs}(\dist)}
\]
Unfortunately, the gap on the right hand side can be very large:
\begin{theorem}\label{thm:cost-gap-example}
 There exists a unit-demand value distribution $\dist$ over $m$ items and a cost vector $\cs\in\R^m$, such that $\buymanyprofit_\cs=\Omega(m)\sprofit_\cs(\dist)$.
\end{theorem}

We instead provide a bi-criteria approximation for the Lagrangian problem. In particular, we compare the profit of the optimal item pricing for cost vector $\cs$ with the profit of the optimal buy many mechanism with production costs $2\cs$. This suffices to imply Theorem~\ref{thm:constrained-alloc-approx} with a slight worsening in the approximation factor.

\begin{theorem}\label{buy-many single buyer with costs}
For any single buyer with value distribution $\dist$ over $m$-items and production costs vector $\cs$,
    \begin{equation*}
    \buymanyprofit_{2\cs}(\dist)\leq 2\ln4m\,\sprofit_{\cs}(\dist).
    \end{equation*}
\end{theorem}

The rest of the section is organized as follows. We first show a complete proof of Theorem~\ref{thm:constrained-alloc-approx} based on Theorem~\ref{buy-many single buyer with costs}. We then describe and verify the gap example in Theorem~\ref{thm:cost-gap-example}. The proof of Theorem~\ref{buy-many single buyer with costs} is  deferred to the appendix. Each of these components is self-contained.

\subsection{Proof of Theorem~\ref{thm:constrained-alloc-approx}}
\begin{proof}[Proof of Theorem~\ref{thm:constrained-alloc-approx}]

We first note that $\srev(x)$ is a concave function over $x$ because it optimizes for revenue over random pricings. 
Fix an ex-ante constraint $x^o$ and consider the function $g(x):=\srev(x)-x\cdot\nabla\srev(x^o)$. This function is maximized at $x=x^o$. Furthermore, since $\srev(x)$ is monotone non-decreasing, $\nabla\srev(x^o)\geq\zeros$, and so $c := \nabla\srev(x^o)$ can be thought of as a vector of production costs. 
Since $g(x)$ is exactly the profit of an item pricing for buyer distribution $\dist$ with allocation $x$ and production cost vector $c$ for all items, we know that the random item pricing $p$ that achieves $\srev(x^o)$ is also the optimal item pricing for a \remove{unit-demand} buyer with value distribution $\dist$ and item production costs $\cs$, without any allocation constraint.
\footnote{For any buyer $i$, it is possible that $p$ has allocation $y\leq x^o$, with $\srev(y)=\srev(x^o)$. However, for any item $j$ such that $y_{j}<x^o_j$, since $\srev(y)=\srev(x^o)$, the gradient $\nabla\srev(x^o)$ has value $c_j=0$ on the $j$th component. Thus the profit of item pricing $p$ for the buyer with production costs $\cs$ is still $\srev(x^o)-\cs\cdot x^o$, although the actual allocation is less than $x^o$.} 

Now consider the buy-many profit optimization problem with production costs $2c$. The optimal profit is given by $\buymanyprofit_{2c}$. Restricting attention to buy many mechanisms that satisfy the ex-ante constraint $x^o$, we let $\buymanyprofit_{2c}(x^o)$ denote the optimal profit obtained over that set of mechanisms. Then, we can apply Theorem~\ref{buy-many single buyer with costs} to obtain:
\begin{eqnarray*}
    \srev(x^o)&=&\sprofit_{\cs}(x^o)+\cs\cdot x^o\\
    &=&\sprofit_{\cs}(\dist)+\cs\cdot x^o\\
    &\geq&\frac{1}{2\ln 4m}\buymanyprofit_{2c}(\dist)+\cs\cdot x^o\\
    &\geq&\frac{1}{2\ln 4m}\buymanyprofit_{2c}(x^o)+\cs\cdot x^o\\
    &=&\frac{1}{2\ln 4m}(\buymanyprofit_{2c}(x^o)+2\cs\cdot x^o)+\cs\cdot x^o\left(1-\frac{1}{\ln 4m}\right)\\
    &\geq& \frac{1}{2\ln 4m}\buymanyrev(x^o),
\end{eqnarray*}
Here the first line is true by extracting the terms of item costs; the second line is true since $x=x^o$ is optimal for profit under item costs $\cs$; the third line is true by Theorem~\ref{buy-many single buyer with costs}; the fourth line is true since adding an allocation restriction cannot increase profit; the last line is true since $\buymanyprofit_{2c}(x^o)+2\cs\cdot x^o\geq\buymanyrev(x^o)$. This finishes the proof of the theorem.
\end{proof}

\subsection{Proof of Theorem~\ref{thm:cost-gap-example}}

\begin{proof}[Proof of Theorem~\ref{thm:cost-gap-example}]
Let $\cs=(0,2^m,2^m,\cdots,2^m)$. 
For every $j$ such that $2\leq j\leq m$, let $v^{(j)}$ be the following unit-demand value function: $v^{(j)}_1= 2^j$; $v^{(j)}_j=2^m$; $v^{(j)}_k=0$ for $k\not\in\{1,j\}$. In other words, the buyer with value $v^{(j)}$ is only interested in two items $1$ and $j$, with the value for the first item being $2^j$, and that for item $j$ being $2^m$. Consider the following value distribution $\dist$: for every $j$ such that $2\leq j\leq m$, with probability $2^{-j}$, $v=v^{(j)}$; for the remaining probability, $v=\zeros$. Now we analyze $\sprofit_\cs$ and $\buymanyprofit_\cs$. 

For any item pricing $p$, consider its profit contribution from the first item and the rest of the items. Since the buyer's value for the first item forms a geometric distribution, the profit contribution from the first item is upper bounded by the revenue of selling only the first item, which is $O(1)$. For the rest of the items, since when the buyer purchases some item $j>1$, the item must have a price at most $2^m=\cs_j$, this means that the profit contribution of item $j$ is at most 0. Thus $\sprofit_\cs=O(1)$.

Consider the following buy-many mechanism: for every $j\geq 2$, there is a menu entry with allocation $\lambda^{(j)}$ and price $p^{(j)}=2^{j-1}+2^{m-1}$, where $\lambda^{(j)}_1=\lambda^{(j)}_j=0.5$. For any $\lambda\in\udelta$ that is not some $\lambda^{(j)}$, its price is determined by the cheapest way to adaptively purchase it with $\lambda^{(2)}, \cdots, \lambda^{(m)}$. For every buyer of type $v^{(j)}$, she buys lottery $(\lambda^{(j)},p^{(j)})$ in the mechanism with utility 0.\footnote{We can reduce the price of $\lambda^{(j)}$ by some small $\epsilon>0$ to make each buyer type's utility be strictly positive.} Since item $j$ is only of interest to buyer $v^{(j)}$, the buyer would not purchase any other set of lotteries in the mechanism. Since profit from buyer $v^{(j)}$ is $2^{j-1}$, and she gets realized in $\dist$ with probability $2^{-j}$, the expected profit of the buy-many mechanism is $\Omega(m)$. 
\end{proof}

\section{Approximation via Sequential Item Pricing}
\label{sec:approx}

We will now focus on item pricings and prove Theorem~\ref{thm:seq-pricing-approx-exante}. In particular, we show that non-adaptive sequential item pricings can obtain half of the ex-ante optimal item pricing revenue, regardless of the order in which buyers are served. 

\begin{numberedtheorem}{\ref{thm:seq-pricing-approx-exante}} {\em (Restatement)}
For any joint distribution $\dists=(\dist_1,\dist_2,\cdots,\dist_n)$ over $n$ unit-demand or additive buyers and any order $\sigma$ on arrival of buyers, there exists a deterministic sequential item pricing $q$ with buyers arriving in order $\sigma$, such that
   \begin{equation*}
       \rev_{\dist}(q)\geq \frac{1}{2}  \exantesrev(\dists).
    \end{equation*}
\end{numberedtheorem}

For additive buyers, the items impose no externalities on each other, and so the theorem follows immediately from the single-unit prophet inequality, Henceforth we focus on unit-demand buyers. Our argument is based loosely around online contention resolution schemes (OCRS) \cite{feldman2016online} and prophet inequality arguments. The idea is to start with the optimal solution to the ex ante item pricing revenue: $x^* := \argmax_{x: \sum_i x_{ij}\le 1 \forall j\in [m]} \sum_i \srev(\dist_i, x_i)$. Then, given some ordering $\sigma$ over the buyers, we try to mimic this allocation by choosing pricings for each buyer that ensure that the buyer receives allocation comparable to $x^*_i$. As in OCRS, we tradeoff assigning enough allocation to a buyer with maintaining a good probability that items remain available for future buyers. 

A key difference in our setting relative to work on OCRS is that the latter mostly focuses on utilitarian objectives, e.g. social welfare, so that the tradeoff is easily quantified: choosing an alternative with half the probability of the ex ante optimum, for example, provides half its contribution to the objective. \citet{chawla2007algorithmic} show how to apply this approach to revenue for unit demand buyers with values independent across items by transforming values to Myersonian virtual values. Unfortunately this approach does not extend to values correlated across items because in correlated settings it is not possible to assign virtual values to each individual value independent of other values.

We develop an alternate argument. For any single unit-demand buyer, we consider how the item pricing revenue changes as the allocation of the buyer is decreased from some intended allocation $x^*$ to a new allocation $y$ that is component-wise smaller. We show that if $x^*$ is realized by item pricing $p$, then we can realize allocation $y$ while obtaining revenue at least $y\cdot p$. In particular, uniformly scaling down allocations by some factor scales down revenue by no more than the same factor. This allows us to carry out the OCRS-style argument. We formalize the above claim as a lemma before providing a proof of Theorem~\ref{thm:seq-pricing-approx-exante}. 

In the following discussion, for any unit-demand buyer with value distribution $\dist$, let $\alloc_{p,S}(\dist)$ be the allocation vector of item pricing $p$ over the set of available items $S\subseteq [m]$. We remove the distribution $\dist$ whenever it is clear from the context. When $S=[m]$, we use $\alloc_p$ instead of $\alloc_{p,[m]}$.
\begin{lemma}\label{lemma:item-pricing-with-less-alloc-exists}

For any unit-demand buyer, any deterministic item pricing $p$, and any distribution over set $S$ of available items, let $x^*=\E_{S}[\alloc_{p,S}]$ be the expected allocations of $p$ conditioned on the available set of items being $S$. Then for any allocation vector $y\in\udelta$ such that $y\preceq x^*$,
there exists a random item pricing $q$ such that 
\begin{equation*}
    \E_{q, S}[\alloc_{q}(S)]= y,\textrm{ and } \rev_{q,S} = y\cdot p.
\end{equation*}

\end{lemma}

\begin{proof}
We first prove the theorem when both $S$ and $p$ are deterministic. Then, we extend the proof to the case where the available set $S$ can be possibly randomized.

For $p=(p_1,\ldots, p_m)\in\R_+^m$ being a deterministic item pricing, assume that the set of available items is fixed to be some deterministic set $S$. For any set $T\subseteq S$, define item pricing $p_T$ to be the pricing $p$ restricted to items in $T$. In other words, $p_{T,j}=p_j$ for all $j\in T$, and $p_{T,j}=\infty$ otherwise. For ease of notation, let $x_T=\alloc_{p,T}$ be the allocation under available set $T$. Observe that for any $j\not\in T$, $x^*_j=0$ and for any $j\in T$, $x_T(j)\geq x^*_j$: the latter is true since under the same pricing, when fewer items are available, buyer types that purchase an item not in $T$ may switch to purchase some item in $T$, while the other buyer types' incentives remain unchanged. Thus for $y\preceq x^*$, $y$ is in the convex hull of the set of $2^{|S|}$ points $X=\{x_T|T\subseteq S\}$. Write $y=\sum_{T\subseteq S}\alpha_T x_{T}$ as a convex combination of vectors in $X$, here $\alpha_T\in[0,1]$ for every $T$, and $\sum_T\alpha_T=1$. Consider the following randomized item pricing $q$: with probability $\alpha_T$, $q=p_T$, $\forall T\subseteq S$. Then the expected allocation of $q$ is exactly $y$. On the other hand, the expected revenue is $\rev_{q,S}=p\cdot y$ because whenever item $j$ is sold in $q$, it is sold at a price of $p_j$.
    
Now let $S$ be a random variable over sets of available items. By defining $x_T=\E_{S}\alloc_{p,T\cap S}$ to be the expected allocation of item pricing $p$ under available item set $T\cap S$, the above proof still goes through. This finishes the proof of the lemma. 
\end{proof}

With the help of the lemma above, we are now ready to prove Theorem~\ref{thm:seq-pricing-approx-exante}.

\begin{proof}[Proof of Theorem~\ref{thm:seq-pricing-approx-exante}]
For every buyer $i$, let $p_i$ be the (randomized) item pricing defining $\exantesrev(\dist)$, and $x_i$ be the corresponding allocation vector. We may assume $\sigma_i=i$ without loss of generality. We build the desired deterministic sequential item pricing $q$ incrementally for every buyer. At arrival of buyer $i$, we use the allocation vector $x_{i}$, the distribution over currently available items (where the randomness comes from previous buyers and item pricings) and \Cref{lemma:item-pricing-with-less-alloc-exists} to produce a (randomized) item pricing vector $q_{i}$ with the property that for $S_i$ being the random set of available items to buyer $i$, 
\begin{equation}\label{eqn:exante}
\E_{q_i, S_i}[x_{q_{i}, S_i}(\dist_{i})] = \frac{x_{i}}{2}, \, \text{and, }\,
    \E_{q_i, S_i}[\rev_{q_{i}}(\dist_{i},S_i)] = \frac{x_{i}}{2}\cdot p_{i}=\frac{1}{2}\exantesrev(\dist_{i}).
\end{equation} 
Here $\rev_{q_i}(\dist_i,S_i)$ denotes the revenue of item pricing $q_i$ for buyer $i$ conditioned on the available item set being $S_i$ when $i$ arrives. If such $q_i$ exists for every $i$, then the random sequential item pricing $q$ satisfies $\rev_{\dists}(q)\geq \frac{1}{2}  \exantesrev(\dists)$. Thus there must exist a realization of $q$ being a deterministic sequential item pricing satisfying the requirement of the theorem.

Now it suffices to show that item pricing $q_i$ exists for \eqref{eqn:exante}, and the rest of the proof is dedicated to proving this. An important observation is that if by induction the item prices $q_{i'}$ satisfying \eqref{eqn:exante} exist for every $i'<i$, then every item belongs to $S_i$ with probability at least $\frac{1}{2}$. Such a observation is true by noticing that by union bound, the allocation of item $j$ in the first $i-1$ steps is at most $\sum_{i'<i}\frac{x_{i'j}}{2}\leq \sum_{i'\leq m}\frac{x_{i'j}}{2}\leq \frac{1}{2}$ since item $j$ has a total ex-ante allocation at most 1.

When $p_i$ is deterministic, since every element in $[m]$ exists with probability at least $\frac{1}{2}$ in $S_i$, we know that for any realized buyer type, her favorite item still remains with probability at least $\frac{1}{2}$ and she would not deviate to purchase something else. Thus $\E_{p_i,S_i}[\alloc_{p_i,S_i}]\succeq \frac{1}{2}x_i$. By \Cref{lemma:item-pricing-with-less-alloc-exists}, there exists a random item pricing $q$, such that $\E_{q,S_i}[\alloc_{q,S_i}]=\frac{1}{2}x_i$, and $\E_{q,S_i}[\rev_{q,S_i}]=\frac{1}{2}x_i\cdot p_i=\frac{1}{2}\exantesrev(\dist_i)$. Thus \eqref{eqn:exante} is satisfied.

When $p_i$ is random, consider any instantiation of $p_i$. The same as the reasoning in the previous paragraph, $\E_{p_i,S_i}[\alloc_{p_i,S_i}]\succeq \frac{1}{2}\alloc_{p_i}$ still holds, and there exists a random item pricing $q_{p_i}$ such that $\E_{q_{p_i},S_i}[\alloc_{q_{p_i},S_i}]=\frac{1}{2}\alloc_{p_i}$, and $\E_{q_{p_i},S_i}[\rev_{q_{p_i},S_i}]=\frac{1}{2}\alloc_{p_i}\cdot {p_i}$. Consider the following random item pricing $q_i$: firstly generate a realization of random item pricing $p_i$, then generate a realization of random item pricing $q_{p_i}$ defined above. The expected allocation of $q_i$ is 
\begin{equation*}
    \E_{p_i}\E_{q_{p_i},S}[\alloc_{q_{p_i},S}]=\frac{1}{2}x_i,
\end{equation*}
while the expected revenue is
\begin{equation*}
    \E_{p_i}\E_{q_{p_i},S}[\rev_{q_{p_i},S}]=\E_{p_i}\left[\frac{1}{2}\alloc_{p_i}\cdot p_i\right]=\frac{1}{2}\exantesrev(\dist_i).
\end{equation*}
Thus $q_i$ satisfies \eqref{eqn:exante}, which finishes the proof of the theorem.

\end{proof}



\section{Discussion}\label{discussion}
In this section, we mention some examples that motivate definition and further work in multi-buyer buy-many mechanisms. Moreover, we discuss why extension of our results to any valuation function is challenging.

\textbf{Motivating Examples. }
Consider a seller who is selling multiple items in multiple markets, and faces a common supply constraint across these markets. The seller is free to choose a different selling mechanism in each market, but interacts with a buyer from each market in just one go. For example, imagine Amazon selling rare books in multiple markets (US, Europe, India). Within each market, Amazon will display a price schedule that does not change frequently based on purchase decisions in other markets (but updates availability). This price schedule could be different for different markets based on local preferences and demand. Within each market individually, given that prices will remain static over short periods of time, a buy-many constraint is a natural property to satisfy. This scenario fits directly within our model.

For another similar scenario consider a travel website like hotwire.com which offers deals on airline tickets, hotel rooms, etc., without revealing complete vendor information. In effect, it sells lotteries. This is another example with supply constraints where the mechanism may personalize prices for each potential buyer (e.g. based on which browser the buyer is using). If the seller uses a non-buy-many mechanism (e.g. if lotteries on multiple items are generally more expensive than individual prices on the items they contain) it would lose customers over time.

\textbf{Difficulties of Extending the Results to All Valuations. }One component of our argument, namely approximating the ex-ante buy many revenue by the ex-ante SRev holds for every possible value function. However, we don’t know how to extend the second part of the argument – approximating the ex-ante SRev using sequential item pricing – for value functions that are not unit-demand or additive. This requires constructing a multi-dimensional prophet inequality. The key technical challenge for non-unit-demand valuations is in keeping track of and controlling how the probability that a particular subset of items is available to an agent depends on decisions of other buyers.
\bibliographystyle{plainnat}
\bibliography{references}

\begin{thebibliography}{27}
\providecommand{\natexlab}[1]{#1}
\providecommand{\url}[1]{\texttt{#1}}
\expandafter\ifx\csname urlstyle\endcsname\relax
  \providecommand{\doi}[1]{doi: #1}\else
  \providecommand{\doi}{doi: \begingroup \urlstyle{rm}\Url}\fi

\bibitem[Alaei(2014)]{alaei2014bayesian}
Saeed Alaei.
\newblock Bayesian combinatorial auctions: Expanding single buyer mechanisms to
  many buyers.
\newblock \emph{SIAM Journal on Computing}, 43\penalty0 (2):\penalty0 930--972,
  2014.

\bibitem[Alaei et~al.(2013)Alaei, Fu, Haghpanah, and Hartline]{alaei2013simple}
Saeed Alaei, Hu~Fu, Nima Haghpanah, and Jason Hartline.
\newblock The simple economics of approximately optimal auctions.
\newblock In \emph{2013 IEEE 54th Annual Symposium on Foundations of Computer
  Science}, pages 628--637. IEEE, 2013.

\bibitem[Babaioff et~al.(2014)Babaioff, Immorlica, Lucier, and
  Weinberg]{babaioff2014simple}
Moshe Babaioff, Nicole Immorlica, Brendan Lucier, and S~Matthew Weinberg.
\newblock A simple and approximately optimal mechanism for an additive buyer.
\newblock In \emph{2014 IEEE 55th Annual Symposium on Foundations of Computer
  Science}, pages 21--30. IEEE, 2014.

\bibitem[Briest et~al.(2015)Briest, Chawla, Kleinberg, and
  Weinberg]{briest2015pricing}
Patrick Briest, Shuchi Chawla, Robert Kleinberg, and S~Matthew Weinberg.
\newblock Pricing lotteries.
\newblock \emph{Journal of Economic Theory}, 156:\penalty0 144--174, 2015.

\bibitem[Cai and Zhao(2017)]{cai2017simple}
Yang Cai and Mingfei Zhao.
\newblock Simple mechanisms for subadditive buyers via duality.
\newblock In \emph{Proceedings of the 49th Annual ACM SIGACT Symposium on
  Theory of Computing}, pages 170--183, 2017.

\bibitem[Chawla and Miller(2016)]{chawla2016mechanism}
Shuchi Chawla and J~Benjamin Miller.
\newblock Mechanism design for subadditive agents via an ex ante relaxation.
\newblock In \emph{Proceedings of the 2016 ACM Conference on Economics and
  Computation}, pages 579--596, 2016.

\bibitem[Chawla et~al.(2007)Chawla, Hartline, and
  Kleinberg]{chawla2007algorithmic}
Shuchi Chawla, Jason~D Hartline, and Robert Kleinberg.
\newblock Algorithmic pricing via virtual valuations.
\newblock In \emph{Proceedings of the 8th ACM Conference on Electronic
  Commerce}, pages 243--251, 2007.

\bibitem[Chawla et~al.(2010)Chawla, Hartline, Malec, and
  Sivan]{chawla2010multi}
Shuchi Chawla, Jason~D Hartline, David~L Malec, and Balasubramanian Sivan.
\newblock Multi-parameter mechanism design and sequential posted pricing.
\newblock In \emph{Proceedings of the forty-second ACM symposium on Theory of
  computing}, pages 311--320, 2010.

\bibitem[Chawla et~al.(2015)Chawla, Malec, and Sivan]{chawla2015power}
Shuchi Chawla, David Malec, and Balasubramanian Sivan.
\newblock The power of randomness in bayesian optimal mechanism design.
\newblock \emph{Games and Economic Behavior}, 91:\penalty0 297--317, 2015.

\bibitem[Chawla et~al.(2019)Chawla, Teng, and Tzamos]{chawla2019buy}
Shuchi Chawla, Yifeng Teng, and Christos Tzamos.
\newblock Buy-many mechanisms are not much better than item pricing.
\newblock In \emph{Proceedings of the 2019 ACM Conference on Economics and
  Computation}, pages 237--238, 2019.

\bibitem[Chawla et~al.(2020)Chawla, Teng, and Tzamos]{chawla2020menu}
Shuchi Chawla, Yifeng Teng, and Christos Tzamos.
\newblock Menu-size complexity and revenue continuity of buy-many mechanisms.
\newblock In \emph{Proceedings of the 21st ACM Conference on Economics and
  Computation}, pages 475--476, 2020.

\bibitem[Chawla et~al.(2021)Chawla, Rezvan, Teng, and
  Tzamos]{chawla2021pricing}
Shuchi Chawla, Rojin Rezvan, Yifeng Teng, and Christos Tzamos.
\newblock Pricing ordered items.
\newblock \emph{arXiv preprint arXiv:2106.04704}, 2021.

\bibitem[Dütting et~al.(2020)Dütting, Kesselheim, and Lucier]{duettinglog}
Paul Dütting, Thomas Kesselheim, and Brendan Lucier.
\newblock An o(log log m) prophet inequality for subadditive combinatorial
  auctions.
\newblock In \emph{2020 IEEE 61st Annual Symposium on Foundations of Computer
  Science (FOCS)}, pages 306--317, 2020.
\newblock \doi{10.1109/FOCS46700.2020.00037}.

\bibitem[Feldman et~al.(2016)Feldman, Svensson, and
  Zenklusen]{feldman2016online}
Moran Feldman, Ola Svensson, and Rico Zenklusen.
\newblock Online contention resolution schemes.
\newblock In \emph{Proceedings of the twenty-seventh annual ACM-SIAM symposium
  on Discrete algorithms}, pages 1014--1033. SIAM, 2016.

\bibitem[Feng and Hartline(2018)]{feng2018end}
Yiding Feng and Jason~D Hartline.
\newblock An end-to-end argument in mechanism design (prior-independent
  auctions for budgeted agents).
\newblock In \emph{2018 IEEE 59th Annual Symposium on Foundations of Computer
  Science (FOCS)}, pages 404--415. IEEE, 2018.

\bibitem[Feng et~al.(2019)Feng, Hartline, and Li]{feng2019optimal}
Yiding Feng, Jason~D Hartline, and Yingkai Li.
\newblock Optimal auctions vs. anonymous pricing: Beyond linear utility.
\newblock In \emph{Proceedings of the 2019 ACM Conference on Economics and
  Computation}, pages 885--886, 2019.

\bibitem[Feng et~al.(2020)Feng, Hartline, and Li]{feng2020simple}
Yiding Feng, Jason Hartline, and Yingkai Li.
\newblock Simple mechanisms for non-linear agents.
\newblock \emph{arXiv preprint arXiv:2003.00545}, 2020.

\bibitem[Feng et~al.(2021)Feng, Hartline, and Li]{feng2021revelation}
Yiding Feng, Jason~D Hartline, and Yingkai Li.
\newblock Revelation gap for pricing from samples.
\newblock In \emph{Proceedings of the 53rd Annual ACM SIGACT Symposium on
  Theory of Computing}, pages 1438--1451, 2021.

\bibitem[Hart and Nisan(2019)]{hart2013menu}
Sergiu Hart and Noam Nisan.
\newblock Selling multiple correlated goods: Revenue maximization and menu-size
  complexity.
\newblock \emph{Journal of Economic Theory}, 183:\penalty0 991--1029, 2019.

\bibitem[Iwasaki et~al.(2010)Iwasaki, Conitzer, Omori, Sakurai, Todo, Guo, and
  Yokoo]{IwasakiCOSTGY10}
Atsushi Iwasaki, Vincent Conitzer, Yoshifusa Omori, Yuko Sakurai, Taiki Todo,
  Mingyu Guo, and Makoto Yokoo.
\newblock Worst-case efficiency ratio in false-name-proof combinatorial auction
  mechanisms.
\newblock In Wiebe van~der Hoek, Gal~A. Kaminka, Yves Lesp{\'{e}}rance, Michael
  Luck, and Sandip Sen, editors, \emph{9th International Conference on
  Autonomous Agents and Multiagent Systems {(AAMAS} 2010), Toronto, Canada, May
  10-14, 2010, Volume 1-3}, pages 633--640. {IFAAMAS}, 2010.

\bibitem[Kleinberg and Weinberg(2012)]{kleinberg2012matroid}
Robert Kleinberg and Seth~Matthew Weinberg.
\newblock Matroid prophet inequalities.
\newblock In \emph{Proceedings of the forty-fourth annual ACM symposium on
  Theory of computing}, pages 123--136, 2012.

\bibitem[Li and Yao(2013)]{li2013revenue}
Xinye Li and Andrew Chi-Chih Yao.
\newblock On revenue maximization for selling multiple independently
  distributed items.
\newblock \emph{Proceedings of the National Academy of Sciences}, 110\penalty0
  (28):\penalty0 11232--11237, 2013.

\bibitem[Ma and Simchi{-}Levi(2021)]{MaS21}
Will Ma and David Simchi{-}Levi.
\newblock Reaping the benefits of bundling under high production costs.
\newblock In Arindam Banerjee and Kenji Fukumizu, editors, \emph{The 24th
  International Conference on Artificial Intelligence and Statistics, {AISTATS}
  2021, April 13-15, 2021, Virtual Event}, volume 130 of \emph{Proceedings of
  Machine Learning Research}, pages 1342--1350. {PMLR}, 2021.

\bibitem[Rubinstein and Weinberg(2018)]{rubinstein2018simple}
Aviad Rubinstein and S~Matthew Weinberg.
\newblock Simple mechanisms for a subadditive buyer and applications to revenue
  monotonicity.
\newblock \emph{ACM Transactions on Economics and Computation (TEAC)},
  6\penalty0 (3-4):\penalty0 1--25, 2018.

\bibitem[Yan(2011)]{yan2011mechanism}
Qiqi Yan.
\newblock Mechanism design via correlation gap.
\newblock In \emph{Proceedings of the twenty-second annual ACM-SIAM symposium
  on Discrete Algorithms}, pages 710--719. SIAM, 2011.

\bibitem[Yao(2014)]{yao2014n}
Andrew Chi-Chih Yao.
\newblock An n-to-1 bidder reduction for multi-item auctions and its
  applications.
\newblock In \emph{Proceedings of the Twenty-Sixth Annual ACM-SIAM Symposium on
  Discrete Algorithms}, pages 92--109. SIAM, 2014.

\bibitem[Yokoo et~al.(2004)Yokoo, Sakurai, and Matsubara]{Makoto04}
Makoto Yokoo, Yuko Sakurai, and Shigeo Matsubara.
\newblock {The effect of false-name bids in combinatorial auctions: new fraud
  in internet auctions}.
\newblock \emph{Games and Economic Behavior}, 46\penalty0 (1):\penalty0
  174--188, January 2004.

\end{thebibliography}

\appendix
\section{Proof of Theorem~\ref{buy-many single buyer with costs}}

\begin{numberedtheorem}{\ref{buy-many single buyer with costs}}
For any single buyer
with value distribution $\dist$ over $m$-items and production costs vector $\cs$,
    \begin{equation*}
    \buymanyprofit_{2\cs}(\dist)\leq 2\ln4m\,\sprofit_{\cs}(\dist).
    \end{equation*}
\end{numberedtheorem}

\begin{proof}[Proof of Theorem~\ref{buy-many single buyer with costs}]
    The proof is similar to the proof of Theorem 1.3 in \cite{chawla2019buy} without production costs. When constructing the item pricing, we need to take the item costs into consideration. 
    
    Let $p$ be the optimal buy-many mechanism under production cost $2\cs$ \footnote{Note that optimal buy-many mechanism is defined over all possible random allocations.}. Observe that for any lottery $\lambda\in\udelta$, without loss of generality we can assume that $p(\lambda)\geq 2\cs\cdot\lambda$: otherwise, we can just remove all menu options with price smaller than production cost, and replace them with the cheapest way to generate the lottery with the remaining options. This way, each lottery $\lambda\in\udelta$ has price at least $2\cs\cdot\lambda$. Furthermore, each buyer type that previously purchase some lottery with production cost larger than the price would deviate to purchase some menu option with price being at least the production cost, while other buyer types' incentives remain unchanged, which means the total profit (i.e. payment minus item costs) does not decrease.

    
    Consider the following item pricing $q$ from $p$: For each item $j$, let
    \begin{equation}\label{eqn:def-of-q}
        q_j=\min_{\lambda\in\udelta}\frac{p(\lambda)-\cs\cdot\lambda}{\lambda_j}+c_j.
    \end{equation}
    In other words, item pricing $q$ is defined by setting the price of each item $j$ to be the cheapest way to purchase the item, while taking into consideration the production cost. We further define item pricing $q_\alpha$ parameterized by $\alpha\in\R$ as a partial scaling of the pricing $q$, where the first term is scaled by the factor $1-\alpha$ but the cost component is left unchanged: 
    \begin{equation*}
        q_{\alpha,j}=\alpha c_j+(1-\alpha)q_j = (1-\alpha) \min_{\lambda\in\udelta}\frac{p(\lambda)-\cs\cdot\lambda}{\lambda_j}+c_j.
    \end{equation*}
    The price of any lottery $\lambda$ under an item pricing $q$ is defined as $q(\lambda)= q\cdot \lambda$. Under item pricing $q_\alpha$, for any buyer with value function $v$, her utility $u(v,\alpha)$ is
    \begin{equation*}
        u(v,\alpha)=\max_{\lambda\in\udelta}\left(v(\lambda)-\alpha \cs\cdot\lambda-(1-\alpha)q(\lambda)\right).
    \end{equation*}
    Let $\lambda_\alpha$ be the lottery purchased by buyer $v$ under pricing $q_\alpha$. By the envelope theorem, we have
    \begin{equation*}
        \frac{d}{d\alpha}u(v,\alpha)=(q-c)\cdot \lambda_\alpha.
    \end{equation*}
    Thus the profit $\profit_{q_\alpha,\cs}(v)$ of buyer $v$ under pricing $q_\alpha$ and production cost $\cs$ is 
    \begin{equation}\label{eqn:rev-util}
        \profit_{q_\alpha,\cs}(v)=q_\alpha(\lambda_\alpha)-c\cdot\lambda_{\alpha}=(1-\alpha)(q-c)\cdot\lambda_\alpha=(1-\alpha)\frac{d}{d\alpha}u(v,\alpha).
    \end{equation}
    Next, we consider two special cases: $\alpha=-1$ and $\alpha=1-\frac{1}{2m}$, and we bound the corresponding utility of buyer $v$ under pricing $q_\alpha$. 

    
    When $\alpha=-1$, we have 
    \begin{eqnarray*}
    q_{-1,j}&=&-c_j+2q_j\\
    &=&c_j+2\cdot\min_{\lambda\in\udelta}\frac{p(\lambda)-\cs\cdot\lambda}{\lambda_j}\\
    &\geq&c_j+2\cdot\min_{\lambda\in\udelta}\frac{p(\lambda)}{2\lambda_j}\\
    &\geq&\min_{\lambda\in\udelta}\frac{p(\lambda)}{\lambda_j}.
    \end{eqnarray*}
    Here the second line is by the definition \eqref{eqn:def-of-q} of $q_j$; the third line is by the earlier observation that $p(\lambda)\geq2\cs\cdot\lambda$ for any lottery $\lambda\in\udelta$. Notice that under buy-many pricing $p$, for any lottery $\lambda^o$, the following strategy gives allocation $\lambda^o\in\udelta$ with price $\sum_{j}\lambda^o_j\min_{\lambda\in\udelta}\frac{p(\lambda)}{\lambda_j}$: draw a set from the lottery $\lambda^o$, and for every item $j$ in this set, repeatedly purchase $\argmin_{\lambda\in\udelta}\frac{p(\lambda)}{\lambda_j}$ until getting the item $j$. Item $j$ is purchased with a total probability of $\lambda_j^o$. Since $p$ satisfies the buy-many constraint, we have for any $\lambda^o\in\udelta$,
    \begin{equation*}
        p(\lambda^o)\leq \sum_{j}\lambda^o_j\min_{\lambda\in\udelta}\frac{p(\lambda)}{\lambda_j}\leq \lambda^o\cdot q_{-1}
    \end{equation*}
    In other words, $q_{-1}$ is more expensive than $p$ on any lottery. Thus the utility of any buyer $v$ is smaller under $q_{-1}$:
    \begin{equation}\label{eqn:ub}
        u_v(p)\geq u(v, -1).
    \end{equation}
    Next, when $\alpha = 1-\frac{1}{2m}$, notice that by \eqref{eqn:def-of-q}, $(q_j-c_j)\lambda_j\leq p(\lambda)-c\cdot \lambda$ for any $\lambda\in\udelta$. Summing over all $j\in[m]$ we have
    \begin{equation}\label{eqn:q-c}
        (q-c)\cdot \lambda=\sum_{j=1}^{m}(q_j-c_j)\lambda_j\leq m(p(\lambda)-c\cdot\lambda).
    \end{equation}
    Let $\lambda^*\in\udelta$ be the lottery purchased by buyer $v$ under the optimal buy-many mechanism $p$ with cost $2\cs$. Then
    \begin{equation}{\label{eqn:lb}}
		\begin{split}
		u\left(v, 1-\frac{1}{2m}\right)&=\max_{\lambda\in\udelta} \left(v(\lambda)-q\cdot\lambda\cdot\frac{1}{2m}-\left(1-\frac{1}{2m}\right)c\cdot\lambda\right)\\
		&\geq \max_{\lambda\in\udelta} \left(v(\lambda)-c\cdot\lambda-\frac{p(\lambda)-c\cdot\lambda}{2}\right)\\
		&=\max_{\lambda\in\udelta}\left(v(\lambda)-p(\lambda)+\frac{p(\lambda)}{2}-\frac{c\cdot\lambda}{2}\right)\\
		&\geq v(\lambda^*)-p(\lambda^*)+\frac{p(\lambda^*)}{2}-\frac{c\cdot\lambda^*}{2}\\
		&\geq u_v(p)+\frac{1}{2}\profit_{p,2\cs}(v),\\
		\end{split}
	\end{equation}
	Here the second line is by \eqref{eqn:q-c}, and the last line is by the definition of the utility $u_v(p)$ of the buyer $v$ under mechanism $p$, and the definition of the profit $\profit_{p,2\cs}(v)$ of buyer $v$ in the setting with production costs. By adding \eqref{eqn:ub} and \eqref{eqn:lb}, we have 
	\begin{eqnarray}\label{eqn:util-diff}
	u\left(v,1-\frac{1}{2m}\right)-u(v,-1)\geq\frac{1}{2}\profit_{p,2c}(v).
	\end{eqnarray}
	Consider $\alpha$ to be randomly drawn from distribution over $[-1,1-\frac{1}{2m}]$ with density $\frac{1}{(1-\alpha)\ln 4m}$. The expected profit from buyer $v$ under pricing $q_\alpha$ and item costs $\cs$ is
	\begin{eqnarray*}
	\E_{\alpha}\profit_{q_\alpha, \cs}(v) &=& \E_{\alpha}(1-\alpha)\frac{d}{d\alpha}u(v,\alpha)\\
	&=&\int_{-1}^{1-\frac{1}{2m}}\frac{1}{(1-\alpha)\ln 4m}\cdot (1-\alpha)\frac{d}{d\alpha}u(v,\alpha) d\alpha\\
	&=&\frac{1}{\ln 4m}\left(u\left(v,1-\frac{1}{2m}\right)-u(v,-1)\right)\\
	&\geq&\frac{1}{2\ln 4m}\profit_{p,2\cs}(v).
	\end{eqnarray*}
	Here the first line is from \eqref{eqn:rev-util}; and the last line is from \eqref{eqn:util-diff}. Taking the expectation over all $v\sim \dist$, we know that there exists a pricing $q_\alpha$ such that $\profit_{q_\alpha,\cs}(\dist)\geq \frac{1}{2\ln 4m}\profit_{p,2\cs}(\dist)$. This finishes the proof of the theorem.
\end{proof}

\end{document}